\documentclass[a4paper,12pt,draftclsnofoot,onecolumn]{IEEEtran}
\usepackage{setspace}
\usepackage{amsmath}
\usepackage{amssymb}
\usepackage{url}
\usepackage{cite}
\usepackage{amsthm}
\usepackage{txfonts}
\usepackage{subfigure}

\allowdisplaybreaks

\usepackage{color}

\usepackage{ifpdf}
\ifpdf
   \usepackage[pdftex]{graphicx}
   \usepackage{epstopdf}
   \AppendGraphicsExtensions{.gif}
\else
   \usepackage[dvips]{graphicx}
\fi

\theoremstyle{plain}
\newtheorem{thm}{Theorem}
\newtheorem{lem}[thm]{Lemma}

\theoremstyle{definition}
\newtheorem{defn}{Definition}

\theoremstyle{remark}
\newtheorem{rem}{Remark}

\begin{document}

\title{Energy Efficient Broadcast in Mobile Networks Subject to Channel Randomness}

\author{Zijie~Zhang,~\IEEEmembership{Member,~IEEE,} Guoqiang~Mao,~\IEEEmembership{Senior Member,~IEEE,} and~Brian~D.~O.~Anderson,~\IEEEmembership{Life~Fellow,~IEEE}%
\IEEEcompsocitemizethanks{
\IEEEcompsocthanksitem This work was mainly undertaken while Z. Zhang was associated with the School of Electrical and Information Engineering, University of Sydney and National ICT Australia. E-mail: zijie.zhang@sydney.edu.au.
\IEEEcompsocthanksitem G. Mao is with the School of Computing and Communications, The University of Technology, Sydney and National ICT Australia. E-mail: g.mao@ieee.org.
\IEEEcompsocthanksitem B. D. O. Anderson is with the Research School of Engineering, Australian National University and National ICT Australia. E-mail:
brian.anderson@anu.edu.au.}%
\thanks{This research is supported by Australian Research Council (ARC) Discovery projects DP110100538 and DP120102030. }%
}

\maketitle

\begin{abstract}
Wireless communication in a network of mobile devices is a challenging and resource demanding task, due to the highly dynamic network topology and the wireless channel randomness. This paper investigates information broadcast schemes in 2D mobile ad-hoc networks where nodes are initially randomly distributed and then move following a random direction mobility model. Based on an in-depth analysis of the popular Susceptible-Infectious-Recovered epidemic broadcast scheme, this paper proposes a novel energy and bandwidth efficient broadcast scheme, named the energy-efficient broadcast scheme, which is able to adapt to fast-changing network topology and channel randomness.  Analytical results are provided to characterize the performance of the proposed scheme, including the fraction of nodes that can receive the information and the delay of the information dissemination process. The accuracy of analytical results is verified using simulations driven by both the random direction mobility model and a real world trace.
\end{abstract}


\section{Introduction}
A mobile ad-hoc network (MANET) is a self-organizing network comprising mobile devices like smart phones, tablet PCs or intelligent vehicles. 
In a MANET, information dissemination often relies on local ad-hoc connections that emerge opportunistically as mobile devices move and meet each other. The main challenge of data communication in a MANET is the time-varying nature of ad-hoc connections, which is attributable to two major factors: dynamic network topology and channel randomness.
%
%

The dynamic topology of a MANET is caused by the mobility of wireless communication devices. Specifically, as mobile users or intelligent vehicles move over time, the distances between mobile devices are changing constantly, resulting in time-varying wireless links. 
This makes many commonly-used routing protocols (e.g. the well-known Ad hoc On Demand Distance Vector (AODV) \cite{Perkins2} or a basic flooding broadcast algorithm \cite{williams-comparison-2002}) less effective for MANETs, because they can only disseminate information to the node(s) that is connected to the source by at least one (multi-hop) path at the time instant when the source node transmits. 
Further, differently from the store-forward pattern that AODV relies on to disseminate information, information dissemination in MANETs has to propagate by way of store-\emph{carry}-forward that allows a node to carry the information over a physical distance and forward to other nodes over time. 
It has been widely recognized that the dynamic topology of a MANET often resembles the topology of a human network \cite{vahdat-epidemic-2000,zhang-on-2011-manet}, in the sense that the movement of mobile devices in a MANET is not only similar to, but often governed by, the movement of their human owners. 
In view of this, epidemic schemes (e.g. \cite{vahdat-epidemic-2000}) have been proposed as a fast and reliable approach to broadcast information in MANETs. 
On the other hand, unlike the spreading of epidemic disease in human networks, information dissemination schemes in MANETs can often be carefully designed to meet certain goals in addition to information delivery, such as achieving energy efficiency or fulfilling certain delay constraint, as shown in this paper.

In addition to the fast-changing network topology, channel randomness also has a significant impact on the performance of information broadcast in MANETs. It has been shown that channel shadowing has negative impacts on information dissemination in MANETs employing traditional routing algorithms like AODV \cite{qin-on-2003}. 
Moreover, the wireless connection between two devices can also be affected by the availability of spectrum resource. Due to scarcity of the radio frequency spectrum, a frequency band is usually shared by more than one service. Consequently, wireless communication between two mobile devices is affected by temporal availability of spectrum band in the vicinity of the devices. Due to these factors, i.e. the uncertainty in the availability of spectrum band as well as the shadowing and fading effects, wireless connections between nodes are random and time-varying, which must be considered in the design of information dissemination schemes for MANETs. Therefore, it is challenging to establish a distributed broadcast scheme for MANETs that is adaptive to both dynamic topology and channel randomness, while meeting the performance objectives, i.e. delay constraint and energy efficiency. 



This paper presents an energy-efficient broadcast scheme for MANETs to address the aforementioned challenges. The design of the broadcast scheme is based on an in-depth analysis of the advantages and inadequacies of the widely-used epidemic broadcast schemes. 
More specifically, the following contributions are made in the paper: 
\begin{enumerate}
\item an energy-efficient broadcast scheme is proposed, motivated by the analysis of the information dissemination process using the Susceptible-Infectious-Recovered (SIR) scheme; 
\item  analytical results are presented on the fraction of nodes that receive the information broadcast by an arbitrary node in a network using the proposed broadcast scheme;
\item it is shown that in comparison with the traditional epidemic schemes, the proposed energy-efficient broadcast scheme consumes less energy and bandwidth while meeting the same performance goal, measured by the fraction of recipients; 
\item the optimal design of the proposed broadcast scheme that achieves the best performance using the minimum amount of energy and bandwidth is proposed;
\item analytical results on the delay required for the information to be received by a pre-designated fraction of nodes in the network are derived; 
\item  different from many existing studies (e.g. \cite{vahdat-epidemic-2000,zhang-on-2011-manet,khabbazian-efficient-2009,chen-information-2010,thedinger-store-2010,friedman-gossiping-2007}), the analysis in this paper takes into account the channel randomness; 
 it is shown (in Sections \ref{section_efficiency} and \ref{section_RDM_simulation}) that shadowing effects improve the performance of the proposed broadcast scheme, measured by the fraction of recipients and delay, which is in contrast to previous studies, e.g. \cite{qin-on-2003}, considering traditional routing algorithms; 
\item the accuracy of analytical results is verified using simulations driven by both the random direction mobility model and a real world trace.
\end{enumerate}

Note that in scenarios where some
network infrastructure is available, i.e. heterogeneous networks,
the communication between devices and base stations can have a great
potential of improving network performance in terms of connectivity,
delay and throughput \cite{lee-on-2013,li-an-2013}.
Though this work focuses on homogeneous ad-hoc networks only, the
analysis can also be useful for further exploration into heterogeneous
networks with infrastructure support \cite{li-multiple-2013,taghizadeh-distributed-2013,bo-mobile-2012},
particularly when infrastructure is introduced to facilitate multi-hop
communications among nodes, or in multi-hop networks with limited
infrastructure support. These networks include mobile social networks
for news/advertisement delivery \cite{ioannidis-optimal-2009}, emergency
communication networks in the event of natural disasters \cite{lien-a-2009},
battlefields communications, emergency rescue and conferences \cite{jeong-a-2014}.

The rest of this paper is organized as follows: Section~\ref{section_related_work} reviews related work. Section~\ref{section_system_model} introduces the system model. 
The analysis of the information dissemination process is presented in Section~\ref{section_analytical}. Section~\ref{section_simulation} validates the analysis using simulations. Finally Section~\ref{section_conclusion} concludes this paper and discusses possible future work.

\section{Related work}\label{section_related_work}
A number of existing studies on information dissemination in wireless networks (e.g. \cite{khabbazian-efficient-2009}) focused on connected networks, where a network is said to be \emph{connected} at a time instant if and only if (iff) there is at least one multi-hop path connecting any pair of nodes. 
In mobile networks, it is often unnecessary or impractical to require that a network is \emph{always} connected \cite{vahdat-epidemic-2000,mao-graph-2009}, due to fast-changing network topology or channel randomness. Hence this paper studies information dissemination in MANETs from a percolation perspective, as described formally in Section \ref{section_analytical}. 

Epidemic broadcast schemes are popularly used for information dissemination in MANETs. In \cite{chen-information-2010}, Chen et al. studied the information dissemination process using a Susceptible-Infectious (SI) epidemic scheme, where every node carries the received information and forwards it to all nodes coming into the radio range. The SI epidemic scheme is a reliable but costly scheme due to a lack of a proper mechanism to stop the transmission. 
Considering a Susceptible-Infectious-Recovered (SIR) epidemic scheme, our previous work \cite{zhang-on-2011-manet} studied the information dissemination process in a MANET. The SIR epidemic scheme postulates that nodes need to keep transmitting 
 for a prescribed time period before recovery (i.e. stopping the transmission); nevertheless a long continuous transmitting period required by the scheme can be difficult or costly to implement in reality. 
Then in a conference version of this work \cite{zhang-opportunistic-2013}, we took shadowing effect into account and proposed the opportunistic broadcast scheme. 
This paper takes a further step by taking fast fading effect into consideration and investigating the optimal design for the opportunistic broadcast scheme that minimises the resource consumption. 

In \cite{clementi-opportunistic-2013}, Clementi
et al. studied the speed of information propagation in a mobile network
where nodes move independently at random over a square area. They obtained an upper bound on the flooding time, which is
the maximum time required for all nodes of the network to be informed.
In \cite{jacquet-information-2009-conf}, Jacquet et al. studied the
information propagation speed in mobile networks where nodes are uniformly
distributed in a bounded area. The nodes were assumed to move following
an i.i.d. random trajectory. An upper bound on the information propagation
speed, viz. ratio of the distance traveled by information
over a given amount of time, was obtained.

Different from the aforementioned two studies
on the information propagation speed (i.e. \cite{clementi-opportunistic-2013,jacquet-information-2009-conf}),
this paper focuses on analysing the fraction of nodes that can receive
the information. We choose to focus on the fraction of recipients
nodes because it is a key performance metric for information broadcast
in a network using the SIR scheme. Compared with the Susceptible-Infectious
(SI) forwarding scheme used in \cite{clementi-opportunistic-2013,jacquet-information-2009-conf},
using the SIR scheme a relay node does \emph{not}
forward the received packets indefinitely, which reflects the real
world scenarios where mobile devices usually have limited energy supply
and/or buffer size. On the other hand, unlike the SI-like schemes,
the SIR scheme does not guarantee that the information can be received
by \emph{all nodes} in a network, and consequently the fraction of recipients
becomes a key performance metric, which is studied in this paper.

Moreover, a number of existing work in this
area (e.g. \cite{clementi-opportunistic-2013,li-rbtp-2013,jacquet-information-2009-conf,khabbazian-efficient-2009})
considered the unit disk model, under which two nodes are directly
connected if and only if the Euclidean distance between them is not
larger than the radio range $r_{0}$. Differently, this work takes
shadowing and fading effects into consideration and shows that channel
randomness can significantly affect the performance of information
dissemination in mobile ad-hoc networks.

%
%

There are other broadcast schemes for MANETs beside epidemic schemes. 
In \cite{friedman-gossiping-2007}, Friedman et al. reviewed some gossip-based algorithms that can be suitable candidates for information dissemination in MANETs. They pointed out that the design of energy and bandwidth efficient information dissemination schemes for MANETs is a challenging and open problem.  A recent work \cite{gao-on-2013} of Gao et al. proposed a novel data forwarding strategy which exploits the transient social contact patterns in social networks. They identified that the low nodal density and the lack of global information are two major challenges in effective data forwarding in delay tolerant networks. We further note that in \cite{Xiang13Energy}, Xiang and Ge  et al. first proposed an energy efficiency model for wireless cellular networks, and on that basis built an analytical relationship among the wireless traffic, wireless channel model and the energy efficiency. Their work suggested that the energy efficiency optimization of wireless networks should consider the wireless channel randomness. In this paper, these challenges are tackled by introducing a distributed energy-efficient broadcast scheme taking into account channel randomness. 

Recently, vehicular ad-hoc networks have been
identified as another valuable application of mobile ad-hoc networks
\cite{karagiannis-vehicular-2011,pagano-is-2013}. Due to the rapid
development of autonomous driving cars, vehicle-to-vehicle communication
network or connected vehicle technology has been attracting increasing interest \cite{luan-enabling-2014}. Due to the relatively slower speed of
infrastructure deployment, in a rather long transitional period, communications
in many of these vehicular networks may receive limited infrastructure
support. Further, the controversial but increasingly popular car-sharing
services \cite{zhang-check-2014} create another future application
for mobile ad-hoc networks. For example, the ad-hoc communication
network formed by vehicles and mobile users can significantly reduce
the requirement and dependance on mobile network operators. In light
of these interesting developments, the simulation section uses vehicle
ad-hoc network as a typical example to evaluate the analytical results
presented in this paper.

\section{System model}\label{section_system_model}
\subsection{Network model}
Suppose that at some initial time instant ($t=0$), a set of $N$ nodes are independently, randomly and uniformly placed on a torus $(0,L]^2$ \cite{franceschetti-random-2007}. It follows that the nodal density is $\lambda=N/L^2$. Then the nodes start to move according to the random direction model (RDM) \cite{camp-a-2002}. Specifically, at time $t=0$, each node chooses its direction independently and uniformly in $[0,2\pi)$, and then moves in this direction thereafter at a constant speed $V$.  It has been shown in \cite{nain-properties-2005} that under the aforementioned model, at any time instant $t\geq 0$, the spatial distribution of nodes is stationary and still follows the uniform distribution. 
Note that the uniform spatial distribution and random direction mobility are both simplified but widely-used models in this field \cite{franceschetti-random-2007,camp-a-2002,nain-properties-2005} to facilitate the analysis. Further, we also evaluate the applicability of our analysis in networks whose nodal distribution and mobility deviate from the above assumption using a real world mobility trace, which is described in detail in Section \ref{section_trace_simulation}.

A commonly-used radio propagation model in this field is the \emph{unit disk model} (UDM), under which two nodes are directly connected iff the Euclidean distance between them is not larger than the radio range $r_0$. More specifically, under the UDM, the received signal strength (RSS) at a receiver separated by distance $x$ from the transmitter is $P_u(x)=Cp_tx^{-\eta}$, where $C$ is a constant, $p_t$ is the transmission power common to all nodes and $\eta$ is the path loss exponent \cite{rappaport-wireless-2002}. A transmission is successful iff the RSS exceeds a given threshold $p_{\min}$. Therefore, the required transmission power $p_t$ allowing a radio range $r_0$ is $p_t=\frac{p_{\min}}{C}r_0^{\eta}$.

The UDM is simple but unrealistic. In reality, the RSS may have significant variations around the mean value; this is typically taken into account in the \emph{log-normal shadowing model} (LSM) \cite{rappaport-wireless-2002}. 
Under the LSM, the RSS attenuation (in dB) follows a Gaussian distribution
, i.e. there is
$10\log_{10}({P_l(x)}/{Cp_tx^{-\eta}})\sim Z$, 
where $P_l(x)$ is the RSS under LSM and $Z$ is the \emph{shadowing factor}, which is a zero-mean Gaussian distributed random variable with standard deviation $\sigma$. When $\sigma=0$, the model reduces to the UDM \cite{rappaport-wireless-2002}. Denote by $q(z)$ the probability density function (pdf) of the shadowing factor $Z$; then:
\begin{equation}\label{eq_pdf_Z}
q(z)=\frac{1}{\sigma \sqrt{2\pi}}\exp\left({-\frac{z^2}{2\sigma ^2}}\right).
\end{equation}

Following common practice in the field \cite{rappaport-wireless-2002}, we consider that the shadowing factors $Z$ between all pairs of transmitter and receiver are independent and all links are symmetric. 
%

In addition to the above large-scale fading, this paper considers a general model of small-scale fading, i.e. the \emph{Nakagami-m} fading model \cite{simon-digital-2000}. Assuming the Nakagami-m fading, the RSS per symbol, $\hat\omega$, is distributed according to a Gamma distribution with the following pdf \cite{simon-digital-2000}: 
\begin{eqnarray}
\textcolor{blue}{\hat\zeta(\hat\omega)=\frac{m^m\hat\omega^{m-1}}{(\mathbb{E}[\hat\omega])^m\Gamma(m)} \exp\left(-\frac{m\hat\omega}{\mathbb{E}[\hat\omega]}\right), ~ ~ \hat\omega\geq 0,}
\end{eqnarray}
where $\Gamma(.)$ is the standard Gamma function and $\mathbb{E}[\hat\omega]= P_l(x)$ is the mean RSS (over time), which is determined by the path loss and shadowing effects. By choosing different values for the parameter $m$, the Nakagami-m fading model easily includes several widely-used fading distributions as its special cases \cite{simon-digital-2000}. 

To incorporate both shadowing and fading effects, we adopt a wireless connection model, named the \emph{general connection model}, which is built on the basis of the random connection model introduced in \cite{franceschetti-random-2007}. 
Specifically, let $P_g(p_t x^{-\eta}, Z, \Omega)$: $\Re^+ \times \Re\times\Re \rightarrow [0, 1]$ be the RSS at a receiver separated by distance $x$ from the transmitter, where $Z$ and $\Omega$ (called \emph{channel factors}) are random variables representing the random variation of the RSS caused by shadowing effect and small-scale fading effect respectively. 
%
Note that the analysis in this paper allows a general form of the RSS function $P_g(p_t x^{-\eta}, Z, \Omega)$, and the analytical results under the log-normal shadowing model with Nakagami fading (called the \emph{Log-normal-Nakagami model}) are provided in Section \ref{section_analytical} as a typical example. 
Lastly, to be practically meaningful in modelling the RSS attenuation, it is assumed that the RSS $P_g(p_t x^{-\eta}, Z, \Omega)$ is a non-decreasing function of $p_t x^{-\eta}$, $Z$ and $\Omega$ respectively.



\subsection{Broadcast scheme}
Suppose that a piece of information is broadcast from an arbitrary node. Once a node receives the information for the first time, it becomes \emph{infectious}. The infectious node holds the information for a fixed amount of time $\tau_s$ (called the \emph{sleep time interval}) followed by a random amount of time $\tauup_r$ (to be described in the next paragraph), then re-transmits the information (to all nodes directly connected to the infectious node) once. Such a \emph{sleep-active cycle} repeats for a fixed number of times, denoted by a positive integer $\beta$, after which the node \emph{recovers}. A {recovered} node stops transmitting the information and will ignore all future transmissions of the same information.  
The information dissemination process naturally stops (i.e. reaches the \emph{steady state}) when there is no infectious node in the network; and the nodes that have received the information are referred to as the \emph{informed nodes}. It is obvious that the fraction of informed nodes is a key performance metric of information dissemination in a given network. Animation of an information dissemination process is available on \cite{my_manet_web}. 

Note that the time interval between two consecutive transmissions is determined by two additive components: a pre-designated sleep time interval $\tau_s$ and a random time interval $\tauup_r$. 
The pre-designated waiting period $\tau_s$ is chosen to allow sufficient time (e.g. $V\tau_s\geq 2r_0$) for a node to move away from the location of its previous broadcast, thereby reducing redundant transmission to the same nodes. 
The random time interval $\tauup_r$ introduces randomness in the transmitting time instants, which can reduce collisions and contention between nodes caused by simultaneous transmissions. 
Further, $\tauup_r$ also reflects the channel access time in some practical scenarios. For example, using carrier sense multiple access (CSMA), if a node finds the channel busy at the end of the pre-designated sleep time interval $\tau_s$, then the node needs to wait a random time interval (viz. random back-off time) for the next transmission opportunity. 
We include $\tauup_r$ in our broadcast scheme to introduce flexibility in determining the transmitting time of each infectious node, so that a node can transmit at its convenience (e.g. when the node using CSMA senses the channel idle) in a decentralized manner while the performance of broadcast in the whole network (e.g. measured by the fraction of informed nodes) is still guaranteed. These features are valuable for a MANET subject to dynamic topology and channel randomness. In view of these features, the above scheme is referred to as an \emph{energy-efficient broadcast scheme}.


Note that when $\tauup_r$ takes a constant value $0$, $\tau_s\rightarrow 0$ and $\beta \tau_s$ is a positive value, the above broadcast scheme becomes a traditional SIR epidemic scheme (c.f. \cite{vahdat-epidemic-2000, zhang-on-2011-manet}) with active period $\beta\tau_s$ seconds. Further, when $\beta\tau_s\rightarrow\infty$, the broadcast scheme becomes a SI epidemic scheme. The advantage of the proposed broadcast scheme compared with the traditional SIR scheme is discussed in detail later in Section \ref{section_efficiency}.

The network described above is denoted by $\mathcal{G}(\eta,\sigma,\lambda,V,r_0,\beta,\tau_s,\tauup_r)$. Further, we assume a sufficiently large network, i.e. $L\gg \beta(\tau_s+\max\{\tauup_{r}\})V$, so that a node will not be wrapped, through its motion in the torus, back to the point where it became infected before it recovers.

\section{Analysis of the information dissemination process}\label{section_analytical}

\subsection{The probability of direct connection}
In order to incorporate channel randomness into the analysis, we need the following lemma.

\begin{lem}\label{lem_rcm}
Consider a set of transmitter-receiver pairs where the wireless channel between each pair has shadowing factor $z$ and fading factor $\omega$.  Then there exists a constant value $r_N(z, \omega)$, called the \emph{radio range given the channel factors $z$ and $\omega$}, such that an arbitrary transmitter-receiver pair in the aforementioned set of transmitter-receiver pairs is directly connected iff the Euclidean distance between the transmitter and receiver is less than or equal to $r_N(z, \omega)$. Further, $r_N(z, \omega)$ is the solution of $r$ to the equation $P_g(p_t r^{-\eta}, z, \omega)=p_{\min}$.
\end{lem}

\begin{proof}
Given the values of the channel factor $z$ and $\omega$, the randomness in the RSS disappears. Noting that $P_g(p_t x^{-\eta}$, $z$, $\omega$) is a non-increasing function of the distance $x$ between transmitter and receiver, the conclusion readily follows. Hence the proof of the lemma is omitted.
\end{proof}


Next, we consider a typical wireless channel model, i.e. the Log-normal-Nakagami model, and obtain the associated value of the radio range.

\begin{lem}\label{lem_r_log}
Suppose that the wireless channel between a transmitter and a receiver is subject to the Log-normal-Nakagami model with channel factors $Z=z$ and $\Omega=\omega$ respectively, then the \emph{radio range given the channel factors} $z$ and $\omega$ is:
\begin{equation}\label{eq_rN}
r_N(z,\omega) \triangleq r_0 \omega^{1/\eta} \exp\left(\frac{z\ln 10}{10\eta}\right),
\end{equation}
where $r_0$ is the radio range under the UDM (c.f. Section \ref{section_system_model}), the pdf of $Z$ is given by Eq. \ref{eq_pdf_Z} and $\Omega$ follows a Gamma distribution with mean $1$.
\end{lem}

\begin{IEEEproof}
Denote by $x$ the distance between a transmitter and a receiver. Subject to Nakagami-m fading, the RSS varies around its mean value over time according to a Gamma distribution with mean $P_l(x)$, where $P_l(x)=Cp_tx^{-\eta}10^{Z/10}$ is the RSS under the log-normal shadowing model.

To facilitate the analysis, we introduce a random variable $\Omega$ which follows a Gamma distribution with mean $1$. Therefore, the pdf of $\Omega$ is 
\begin{eqnarray}\label{eq_pdf_w0}
\zeta(\omega)=\frac{m^m\omega^{m-1}}{\Gamma(m)} \exp(-m\omega), ~ ~ \omega\geq 0.
\end{eqnarray}

Further, it can be shown that for any constant $C_1\in\Re$, the random variable $C_1\Omega$ follows a Gamma distribution with mean $C_1$. Then under the Log-normal-Nakagami model, the RSS at a receiver at distance $x$ from the transmitter is $P_g(p_t x^{-\eta}, Z, \Omega)=P_l(x)\Omega=Cp_tx^{-\eta}10^{Z/10}\Omega$.

Recall that two nodes are directly connected iff the RSS exceeds a given threshold $p_{\min}$. Without shadowing and fading effects, i.e. considering path loss only, there holds $p_{\min}=Cp_tr_0^{-\eta}$ (c.f. Section \ref{section_system_model}). With shadowing and fading, there holds:
\begin{eqnarray}
\Pr\left(P_g(p_t x^{-\eta}, Z, \Omega)\geq p_{\min}\right) 
=  \Pr\left(Cp_tx^{-\eta}10^{Z/10}\Omega \geq Cp_tr_0^{-\eta} \right)
= \Pr\left( x\leq r_0\Omega^{\frac{1}{\eta}}\exp(\frac{Z\ln 10}{10\eta})\right).
\end{eqnarray}

Conditioned on the channel factors $Z=z$ and $\Omega=\omega$ between two nodes, the two nodes are directly connected iff their distance is not larger than $ r_0\omega^{1/\eta}\exp(\frac{z\ln 10}{10\eta})$. 
\end{IEEEproof}




\subsection{The effective node degree}\label{section_degree}

\begin{defn}\label{def1}
The \emph{effective node degree} $R_0$ of an infectious node is the expected number of nodes that are directly connected to the infectious node in at least one of the $\beta$ transmissions.
\end{defn}

Note that $R_0$ is the same for all nodes due to the stationarity and homogeneity of node distribution on the torus. 
To compute the effective node degree, we further need to calculate the clustering factor as defined in the following.


\begin{defn}\label{def_cluster_coefficient}
The \emph{clustering factor} $\phi(\tau_s)$ is the expected number of nodes that are directly connected to an infectious node in \emph{both} of two consecutive transmissions when the sleep time interval is $\tau_s$.
\end{defn}

Using Lemma \ref{lem_r_log}, we have the following results.
\begin{lem}\label{lem_cluster_coefficient}
Consider a network $\mathcal{G}(\eta,\sigma,\lambda,V,r_0,\beta,\tau_s,\tauup_r)$. The clustering factor satisfies
\begin{eqnarray}\label{eq_cluster_coefficient}
&& \hspace{-15px}\phi(\tau_s) \hspace{-1px}
 = \hspace{-5px} \int_{0}^{\infty}\hspace{-4px} \int_{0}^{\infty}\hspace{-4px} \int_{-\infty}^{\infty}\int_{-\infty}^{\infty} \int_0^{\infty} \hspace{-4px} \int_0^{\pi} \hspace{-3px} A_p(\theta,\tau_s,r_N(z_1,\omega_1),r_N(z_2,\omega_2))  \nonumber\\
&& \hspace{-15px}\times \frac{\lambda}{\pi} p_\tau(\tau_r) q(z_1)q(z_2) \zeta(\omega_1)\zeta(\omega_2) d\theta d\tau_r  dz_1dz_2 d\omega_1 d\omega_2,
\end{eqnarray}
where $A_p(\theta,\tau_s, r_1, r_2)$ is given by Eq.~\ref{eq_intersection}.
\end{lem}

\begin{IEEEproof}
Denote by $\Theta$ the angle measured counterclockwise from the moving direction of an infectious node to the moving direction of an arbitrary node. Recall that the direction of a node is randomly and uniformly chosen in $[0, 2\pi)$, independent of the directions of other nodes. Consequently, it can be shown that the angle $\Theta$ is also uniformly distributed in $[0, 2\pi)$.

Suppose that an infectious node transmits once at point $S_1$, then it moves by distance $(\tau_s+\tau_r)V$ to point $S_2$ and transmits again, as shown in Fig.~\ref{pic_RD_coverage_burst}.

\begin{figure}[ht]
\centering\vspace{-5px}
\includegraphics[scale=0.5]{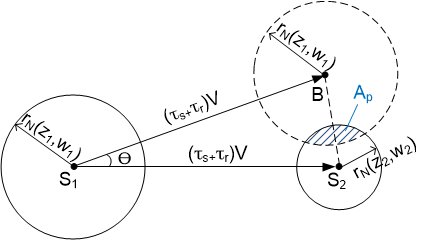}\vspace{-5px}
\caption{An illustration to the nodes (in the shaded area $A_p$) that are directly connected to an infectious node in both two consecutive transmissions (occurring at $S_1$ and $S_2$). Symbols are defined in Lemma~\ref{lem_cluster_coefficient}.}\vspace{-6px}
\label{pic_RD_coverage_burst}
\end{figure}

Next, we focus on a subset of nodes that fulfil the following three conditions: 1) they move in direction $\Theta\in(\theta,\theta+d\theta)$; and 2) their RSS from the infectious node has channel factors $Z_1\in(z_1,z_1+dz_1)$ and $\Omega_1\in(\omega_1,\omega_1+d\omega_1)$ when the infectious node transmits at $S_1$; and 3) their RSS from the infectious node has channel factors $Z_2\in(z_2,z_2+dz_2)$ and $\Omega_2\in(\omega_2,\omega_2+d\omega_2)$ when the infectious node transmits at $S_2$. 
Due to the independence of channel factors and the mobility of nodes, these nodes are uniformly distributed with density $\frac{\lambda}{2\pi} q(z_1)q(z_2)\zeta(\omega_1)\zeta(\omega_2) d\theta dz_1dz_2d\omega_1 d\omega_2$. 
Among this subset of nodes, the nodes that are connected to the infectious node in the first transmission are in the disk centered at point $S_1$ with radius $r_N(z_1,\omega_1)$, denoted by $C(S_1, r_N(z_1,\omega_1))$. Further, when the infectious node transmits at $S_2$, these nodes move by distance $(\tau_s+\tau_r)V$ from being contained in $C(S_1, r_N(z_1,\omega_1))$ to being contained in a new disk $C(B, r_N(z_1,\omega_1))$ \footnote{The directions of nodes can differ by $d\theta$, whose impact on the analysis however becomes vanishingly small when $d\theta\rightarrow 0$ while $(\tau_s+\tau_{r})V$ is finite.} as shown in Fig. \ref{pic_RD_coverage_burst}. Then, the nodes that are connected to the infectious node in both transmissions are in the intersectional area $C(S_2,r_N(z_2,\omega_2))\cap C(B,r_N(z_1,\omega_1))$. Denote by $ A_p(\theta,\tau_s,r_N(z_1,\omega_1),r_N(z_2,\omega_2))$ the size of the intersectional area $C(S_2,r_N(z_2,\omega_2))\cap C(B,r_N(z_1,\omega_1))$. It can be readily calculated using the following formula \cite{mathworld-circle}:
\begin{eqnarray}\label{eq_intersection}
\hspace{-20px}A_p(\theta,\tau_s,r_1,r_2)=
\begin{cases}
\min(\pi r_1^2, \pi r_2^2), ~~\mathrm{for}~\psi(\theta, \tau_s)\leq |r_1-r_2| \\
r_1^2\arccos(\frac{\psi^2(\theta, \tau_s)+r_1^2-r_2^2}{2r_1\psi(\theta, \tau_s)}) +r_2^2\arccos(\frac{\psi^2(\theta, \tau_s)+r_2^2-r_1^2}{2r_2\psi(\theta, \tau_s)}) \cr 
~~~-\frac{1}{2}\hspace{-1px}\sqrt{[(r_1+r_2)^2-\psi^2(\theta, \tau_s)][\psi^2(\theta, \tau_s)-(r_1-r_2)^2]},\\
~~~~\mathrm{  for  } |r_1-r_2|<\psi(\theta, \tau_s)<r_1+r_2 \cr
 0, ~~\mathrm{otherwise},
\end{cases}
\end{eqnarray}
where $\psi(\theta, \tau_s)\hspace{-2px}=\hspace{-2px}2(\tau_s+\tau_r)V\sin\frac{\theta}{2}$ is the length of $BS_2$.


Next we consider all subsets of nodes. Note that only the cases for $\theta\in[0,\pi)$ need to be calculated due to symmetry. Then the clustering factor satisfies $\phi(\tau_s) =  \lambda E[A_p(\theta,\tau_s,r_N(z_1,\omega_1),r_N(z_2,\omega_2))]$.
\end{IEEEproof}

\begin{rem}
The results of Lemma \ref{lem_cluster_coefficient} can be extended to different channel models by substituting the pdfs of the channel factors $z_1, z_2, \omega_1$ and $\omega_2$.
\end{rem}


Finally, we have the following theorem for the value of $R_0$.

\begin{thm}\label{thm_R0}
Consider the network $\mathcal{G}(\eta,\sigma,\lambda,V,r_0, \beta,\tau_s,\tauup_r)$. Under the Log-normal-Nakagami model, the effective node degree satisfies
\begin{equation}\label{eq_R0}
R_0\leq\beta \lambda\pi r_0^2 \exp\left(\frac{(\sigma\ln 10)^2}{50\eta^2}\right) \frac{m^{\frac{-2}{\eta}} \Gamma(m+\frac{2}{\eta})}{\Gamma(m)}- (\beta-1)\phi(\tau_s),
\end{equation}
where $\Gamma(.)$ is the standard Gamma function and $\phi(.)$ is given by Lemma \ref{lem_cluster_coefficient}.
\end{thm}

\begin{IEEEproof}
We first consider two consecutive transmissions, as illustrated in Fig. \ref{pic_RD_coverage_burst}. Along the same lines as Lemma \ref{lem_cluster_coefficient}, consider that the channel factors of the $i^{th}$ transmission (for $i=\{1,2,...,\beta\}$) is $Z_i\in[z_i,z_i+dz_i]$ and $\Omega_i\in[\omega_i,\omega_i+d\omega_i]$. 
Define $\phi_{1,2}(\tau_s)$ to be the clustering factor of the two consecutive transmissions. It is straightforward that $\phi_{1,2}(\tau_s)$ is given by the Eq. \ref{eq_cluster_coefficient} without integrating over $z_1, z_2, \omega_1$ and $\omega_2$.

It is straightforward that the size of the area covered by the radio range in the $i^{th}$ transmission is $\pi(r_N(z_i,\omega_i))^2$. Further, the size of the area wherein the nodes receive both the $1^{st}$ and the $2^{nd}$ transmissions is $\phi_{1,2}(\tau_s)/\lambda$. Denote by $A_R$ the total (union) size of the area wherein the nodes receive either one of the $1^{st}$ or the $2^{nd}$ transmissions. It is straightforward that $A_R=\pi(r_N(z_1,\omega_1) )^2 \hspace{-2px}+\hspace{-2px} \pi(r_N(z_2,\omega_2) )^2 - \phi_{1,2}(\tau_s)/\lambda$.

Then the expected number of nodes that are directly connected to an infectious node in either of the two transmissions can be calculated using Eq.~\ref{eq_pdf_Z} and Lemma \ref{lem_r_log}: 
\begin{eqnarray}\label{eq_R0_3}
&&\hspace{-8px}\int_{0}^{\infty}\hspace{-6px} \int_{0}^{\infty}\hspace{-5px} \int_{-\infty}^{\infty}\hspace{-1px}  \int_{-\infty}^{\infty}\hspace{-5px} \Big(\lambda\pi(r_N(z_1,\omega_1) )^2 \hspace{-2px}+\hspace{-2px} \lambda\pi(r_N(z_2,\omega_2) )^2   
 \hspace{-4px}  - \hspace{-1px} \phi_{1,2}(\tau_s) \hspace{-1px}\Big) q(z_1)q(z_2)\zeta(\omega_1)\zeta(\omega_2) dz_1dz_2d\omega_1 d\omega_2  \nonumber\\
&=&\hspace{-7px}2\int_{0}^{\infty} \hspace{-2px}\int_{-\infty}^{\infty} \hspace{-4px}\lambda\pi \left( r_0\omega^{1/\eta}\exp(\frac{z\ln 10}{10\eta}) \right)^2  \frac{1}{\sigma \sqrt{2\pi}}\exp\left({-\frac{z^2}{2\sigma ^2}}\right)  dz  
\zeta(\omega) d\omega -\phi(\tau_s) \\
&=&\hspace{-7px}2\int_{0}^{\infty} \hspace{-3px} \lambda\pi r_0^2 \omega^{2/\eta} \exp\left(\frac{(\sigma\ln 10)^2}{50\eta^2}\right) \zeta(\omega) d\omega  -\phi(\tau_s) 
=\hspace{-3px}2 \lambda\pi r_0^2 \exp\left(\frac{(\sigma\ln 10)^2}{50\eta^2}\right) \frac{m^{\frac{-2}{\eta}} \Gamma(m+\frac{2}{\eta})}{\Gamma(m)}  -\phi(\tau_s) \nonumber .
\end{eqnarray}


\begin{figure}[ht]
\centering\vspace{-5px}
\includegraphics[scale=0.5]{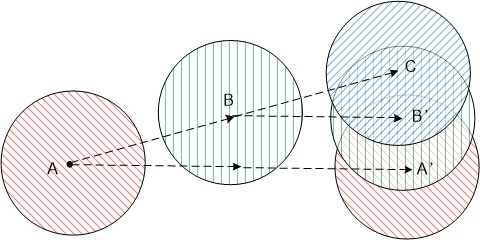}\vspace{-5px}
\caption{An illustration of the area covered by the radio range of an infectious node during three transmissions at points $A$, $B$, and $C$ sequentially. Consider a set of nodes moving in the direction $AA'$. When the infectious node transmits the third time (at point $C$), the nodes that have received the first and second transmissions are in the disks centred at points $A'$ and $B'$ respectively.}\vspace{-6px}
\label{pic_RD_coverage2_burst}
\end{figure}

Using the inclusion-exclusion principle, Eq. \ref{eq_R0} can be readily obtained. More specifically, if the number of transmissions is $\beta=1$, then $R_0$ is equal to the first term in Eq. \ref{eq_R0}. If the number of transmissions is $\beta=2$, then $R_0$ is equal to the two terms in Eq. \ref{eq_R0}. 
Further, when $\beta\geq 3$, the analysis involves the intersectional area of more than two circles, as shown in Fig. \ref{pic_RD_coverage2_burst}. To avoid complicated formulas, we only calculate the intersectional area of two circles (viz. the term $\phi(\tau_s)$) and provide an upper bound on $R_0$ as shown in Eq. \ref{eq_R0}. 
\end{IEEEproof}


\subsection{Percolation probability}\label{section_percolation}
We first study the fraction of informed nodes from a percolation perspective \cite{franceschetti-random-2007} \emph{asymptotically}, viz. we increase the network area towards infinity (i.e. let $L\rightarrow\infty$) while keeping other parameters (i.e. $\lambda,V,r_0,\beta,\tau_s$ and $\tauup_r$) unchanged. 
%

\textcolor{blue}{Consider a realization of a network on a torus
$(0,L]^{2}$. Denote by a constant $N(L)$ the total number of nodes
in this network. Further denote by a random integer $N_{R}(L)$ the
number of informed nodes of a packet (i.e. the nodes that have received
the packet at the end of the information dissemination process) in
this network. Then it is straightforward that the fraction of informed
nodes of the packet in this network can be calculated by $N_{R}(L)/N(L)$.
Note that the fraction of informed nodes is calculated for each information
dissemination process in each realization of a random network.}

\textcolor{blue}{Note that in an infinite network, the network
is said to \emph{percolate} if there exists a component of infinite size
in the network \cite{gilbert-random-1961,dousse-impact-2005,franceschetti-random-2007},
where a \emph{component} is a maximal set of nodes in the network such that
there is an end-to-end path between every pair of nodes in the set.
Percolation probability is the probability that the network percolates.
As widely used in the study of percolation in random networks
\cite{franceschetti-random-2007},
we consider asymptotic networks in this paper. More specifically,
we consider asymptotic networks where $L\rightarrow\infty$ while keeping
the nodal density unchanged. It is straightforward that the total
number of nodes in each realization of the network becomes $N(L)\rightarrow\infty$. }

\textcolor{blue}{Further note that, a series $a_{n}$ depending
on $n$ is said to be \emph{non-vanishingly-small}
if there exists a sufficiently-small positive constant $\varepsilon$
and a positive integer $n_{0}$ such that for all $n>n_{0}$, $a_{n}>\varepsilon$. In this paper, we are interested in the probability that the random
number $N_{R}(L)/N(L)$ is non-vanishingly-small as $N(L)$ approaches
infinity, i.e. the probability that the fraction of informed nodes
is non-vanishingly-small. With a bit twist of its standard definition,
we call this probability the \emph{percolation probability} in our paper. }

\begin{defn}\label{defn_percolation}
The \emph{percolation probability} $p_c$ of a MANET is the probability that a piece of information broadcast from an arbitrary node can be received by a \emph{non-vanishingly-small} fraction of nodes asymptotically.
\end{defn}

The main result of this \textcolor{blue}{subsection} is as follows:

\begin{thm}\label{thm_percolation}
Consider a network $\mathcal{G}(\eta,\sigma,\lambda,V,r_0,\beta,\tau_s,\tauup_r)$ with effective node degree $R_0$. The percolation probability $p_c$ satisfies $p_c\leq 1+\frac{1}{R_0}W(-R_0e^{-R_0})$, where $W(.)$ is the Lambert W Function.
\end{thm}

\begin{IEEEproof}
We first model the information dissemination process using a Galton-Watson branching process \cite{jagers-branching-1975}. The root (viz. the $0^{th}$ generation) is the source node. The expected number of children per node is given by $R_0$. Denote by $\chi(k)$ the number of individuals in the $k^{th}$ generation of the branching process.

Define $q\triangleq \Pr(\lim_{k\rightarrow\infty}\chi(k)\rightarrow 0)$ to be the \emph{extinction probability} of the branching process, viz. the probability that the number of individuals in the $k^{th}$ generation diminishes to zero as $k\rightarrow\infty$. It has been shown in \cite[Theorem 6.5.1]{jagers-branching-1975} and \cite{zhang-on-2011-manet} that the extinction probability $q$ is the smallest non-zero solution of 
$q=\exp((1-q)R_0)$ as the number of nodes $N\rightarrow\infty$. Solving the equation, it can be obtained that $q=\frac{W(-R_0 e^{-R_0})}{-R_0}$, where $W(.)$ is the Lambert W Function~\cite{corless-on-1996}.



Next we establish the connection between the branching process and the information dissemination process in a MANET. 
Denote by $\chi_N(k)$ the number of nodes in the $k^{th}$ generation of the information dissemination process, where the $k^{th}$ generation of the information dissemination process consists of the nodes that receive the information for the first time from a node belonging to the $(k-1)^{th}$ generation, and the $0^{th}$ generation is the source node. 
Similarly as above, define $q_N\triangleq\Pr(\lim_{k\rightarrow\infty} \chi_N(k) \rightarrow 0)$ to be the extinction probability of the information dissemination process. Because some of the $R_0$ nodes that are connected to an infectious node may have already received the information from other infectious nodes, the number of children per node for the information dissemination process is stochastically less than that in the branching process introduced earlier. Therefore, there holds  $\chi_N(k)\preceq\chi(k)$ \footnote{Using stochastic ordering, we say $\chi_N(k)\preceq\chi(k)$ iff $\Pr(\chi_N(k)>c)\leq\Pr(\chi(k)>c)$ for any constant $c$.}. Then there holds  
$
q_N\geq q=\frac{W(-R_0 e^{-R_0})}{-R_0}
$.

Denote by $p'_c$ the probability that the information dissemination process does not go extinct. It is clear that 
$
p'_c=1-q_N
$. 
Further, it is obvious that the information dissemination process does not \textcolor{blue}{become} extinct is a necessary but not sufficient condition for having an non-vanishingly small fraction of informed nodes. Therefore, there holds $p_c\leq p_c'$. 
\end{IEEEproof}

\subsection{Expected fraction of informed nodes}\label{section_fraction}
Define $z_0$ as the expected fraction of informed nodes in the steady state of a percolated network. Then we report the following two results:

\begin{thm}\label{thm_size}
Consider a network ($L\rightarrow\infty$), whose effective node degree is $R_0$. The expected fraction of informed nodes in the steady state of a percolated network satisfies $z_0\leq 1+\frac{1}{R_0}W(-R_0e^{-R_0})$, where $W(.)$ is the Lambert W Function.
\end{thm}

Theorem \ref{thm_size} can be readily proved using a set of ordinary differential equations and a mean field limit theorem that is commonly used in the analysis of epidemic broadcast schemes, c.f. \cite{isham-stochastic-2004, zhang-on-2011-manet} and the reference herein. The proof is therefore omitted here and the accuracy of the result is further verified by simulation in Section \ref{section_simulation}.

Theorems \ref{thm_percolation} and \ref{thm_size} show that the effective node degree determines the performance of broadcast scheme. We next present further analysis quantifying the value of the effective node degree.

\subsection{Energy and bandwidth efficiency}\label{section_efficiency}
We next introduce the energy and bandwidth
consumption metrics. Specifically, assume that the time spent on transmitting
a packet of unit size over a single hop is a constant $T_{t}$.
Therefore the energy consumed in transmitting a packet is $T_{t}p_{t}$,
denoted via a single constant $E_{1}=T_{t}p_{t}$. Similarly, denote by constant
$E_{2}$ the energy consumed when receiving a packet at a single node. 


Denote by random variable $\mathcal{D}_{i}$
the node degree of a randomly chosen node at a random time instant.
Then the sum of the energy consumption for a randomly chosen node
broadcasting a packet and the energy consumed by all its neighbours in
receiving the packet is evidently $E_{1}+\mathcal{D}_{i}E_{2}$. It then follows that the expected energy consumption for a randomly
chosen node broadcasting a packet and all its neighbours receiving
the packet is $E_{1}+\mathbb{E}[\mathcal{D}_{i}]E_{2}$, where $\mathbb{E}[\mathcal{D}_{i}]$
is the average node degree. 

Therefore, we can combine the energy consumption
at a single node, including energy consumed by other nodes in receiving
its packet, by a constant $E_{c}=E_{1}+\mathbb{E}[\mathcal{D}_{i}]E_{2}$. As
manifested in the above equation, the overall energy consumption is
directly related to the number of transmissions and equals to the
number of transmissions times $E_{c}$.

Similarly, if each transmission occupies
(e.g. using CSMA) the frequency band in an area whose expected size
is $A_{c}$, then the expected size of the area where the frequency
band is occupied by an infectious node during its $\beta$ transmissions
is $\beta A_{c}$. Therefore, the consumption of
bandwidth is also an increasing function of the number of transmissions. Therefore, to save energy and bandwidth, we need to reduce the number of transmissions $\beta$.

On the other hand, to meet a pre-designated broadcast performance objective, measured by the percolation probability and expected fraction of informed nodes, a certain number of transmissions are required. Given the dependence of the performance objective on the effective node degree (which is determined by the number of transmissions $\beta$) and the reliance of the energy and bandwidth consumption on the number of transmissions $\beta$, we propose using the following ratio to measure the \emph{energy and bandwidth efficiency} of the proposed broadcast scheme:
\begin{equation}\label{eq_efficiency}
Y\triangleq\frac{R_0}{\beta},
\end{equation} 
which is the average effective node degree achieved per transmission.

To improve energy and bandwidth efficiency, it is obvious that the clustering factor, which characterises the amount of overlap between two transmissions, needs to be reduced. The following lemma reports a useful property of the clustering factor $\phi(\tau_s)$.

\begin{lem}\label{lem_cluster_coefficient_mono}
The clustering factor $\phi(\tau_s)$ is a monotone non-increasing function of $\tau_s$, when all other parameters (i.e. $\eta,\sigma,\lambda,V,r_0,\beta$ and $\tauup_r$,) are fixed.
\end{lem}
\begin{IEEEproof}
In Eq. \ref{eq_cluster_coefficient}, the only term determined by $\tau_s$ is the size of the intersectional area $\int_0^{\pi}A_p(\theta,\tau_s,r_N(z_1,\omega_1),r_N(z_2,\omega_2))d\theta$. It can be shown (easily from Eq. \ref{eq_intersection}) that for any given value of $\theta$, $A_p(\theta,\tau_s,r_N(z_1,\omega_1),r_N(z_2,\omega_2))$ is a non-increasing functions of $\tau_s$, then the conclusion follows that the clustering factor $\phi(\tau_s)$ is a monotone non-increasing function of $\tau_s$.
\end{IEEEproof}

Introducing the sleep time interval $\tau_s$, the proposed broadcast scheme separates consecutive transmissions in both space and time, allowing each transmission to be received by more new nodes that have not received the information during the previous transmission(s). 
This is more energy and bandwidth efficient than the traditional forwarding schemes (e.g. SI or SIR scheme) where an infectious node transmits the information to every susceptible node coming into the radio range. 
To quantitatively compare the energy and bandwidth efficiencies between a traditional SIR \footnote{We do not compare with the SI scheme \cite{chen-information-2010} because it does not have a proper mechanism of stopping transmission; hence its energy and bandwidth consumption cannot be bounded.} scheme and our scheme, we consider a SIR scheme with a perfect neighbour discovery mechanism 
where an infectious node only wakes up when there is a new node coming into the radio range. 
Then the energy and bandwidth efficiency of the SIR scheme can be calculated in the same way as Eq. \ref{eq_efficiency}:
\begin{equation}\label{eq_efficiency2}
\hat{Y}\triangleq\frac{R_0}{R_0-\lambda\pi r_0^2+1},
\end{equation} 
where $\lambda\pi r_0^2$ is the expected number of nodes receiving the first transmission and $R_0-\lambda\pi r_0^2+1$ is the expected number of transmissions required by an infectious node using an ideal SIR scheme to transmit a piece of information to $R_0$ nodes.

As can be seen in Fig. \ref{pic_ana_efficiency_compare}, the proposed scheme has a higher energy and bandwidth efficiency than the SIR scheme, especially when the density is large. This is because when the density is large, the traditional method needs to transmit frequently to every new neighbor whereas the proposed scheme can wait until the infectious node moves to a new region and it will then transmit to a set of new nodes at the same time.


\begin{figure}[ht]
\centering\vspace{-5px}
\includegraphics[scale=0.6]{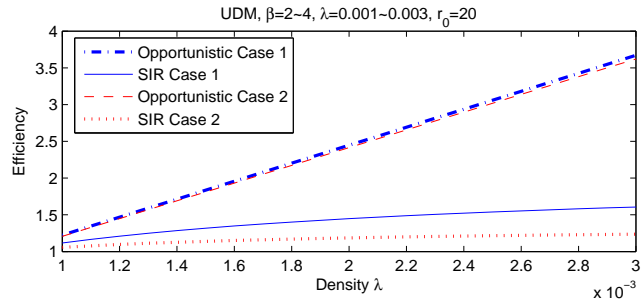}\vspace{-6px}
\caption{The energy and bandwidth efficiency metrics under the proposed scheme (given by Eq. \ref{eq_efficiency}) and the traditional SIR scheme (given by Eq.~\ref{eq_efficiency2}). Note that $R_0$ in Eq.~\ref{eq_efficiency2} takes the same value as the $R_0$ in Eq.~\ref{eq_efficiency}. In Case 1 and Case 2, the proposed scheme transmits $\beta=2$ and $\beta=4$ times respectively. Recall that the efficiency metric is a measure of the average effective node degree achieved per transmission.}\vspace{-6px}
\label{pic_ana_efficiency_compare}
\end{figure}


\subsection{Information dissemination delay}\label{section_delay}
Suppose that a piece of information is broadcast from an arbitrary node at time $t=0$ using the proposed broadcast scheme. Let $T(z)$ be the expected time when the fraction of informed nodes reaches $z$, for $0<z< 1$.

\begin{thm}\label{thm_delay_lower}
Consider a network $\mathcal{G}(\eta,\sigma,\lambda,V,r_0,\beta,\tau_s,\tauup_r)$, whose effective node degree is $R_0$.  In a large network where $L\rightarrow\infty$ and $N=\lambda L^2\rightarrow\infty$, there holds
\begin{equation} \label{eq_delay_lower}
T(z)\geq \tau_{s} \left\lfloor 1+\frac{\ln{Nz}}{\ln(1+\frac{R_0}{\beta})} \right\rfloor .
\end{equation}
\end{thm}

\begin{IEEEproof}
Recall that each infectious node has $\beta$ transmissions separated by a random time interval $\tau_s+\tauup_r$. In this proof, we obtain a lower bound on the delay $T(z)$ by constructing a new network (denoted by $\mathcal{G}_{L}$) where the time interval between any two consecutive transmissions of each node is a constant $\tau_{s}$ and the infectious nodes never recover. Then in the new network $\mathcal{G}_{L}$, all infectious nodes transmit simultaneously at time $\tau_{s},2\tau_{s},3\tau_{s},\dots$. Hence, the delay of information dissemination in the new network is a lower bound on the delay in the original network because 1) each node retransmits the information sooner (sleep time interval is $\tau_{s}$ in the new network compared with $\tau_s+\tauup_r$ in the original network), and 2) the number of transmissions of each node is not smaller than or equal to that in the original network, and 3) the number of infectious nodes created by each infectious node is not smaller than that in the original network (to be described in the next paragraph).

Let $a_k$ be the total number of infectious nodes at time $k\tau_{s}$ in the new network $\mathcal{G}_{L}$, for $k\geq 0$. Initially, there is $a_0=1$. 
Denote by $Q_i\in\{0,1,2,\dots\}$ the number of nodes directly connected to node $i$ at a randomly-chosen time instant. 
Due to uniform distribution of nodes and the fact that nodes move independently of one another, it can be shown that $q_i$ has an identical distribution across all nodes. 
Further, there is $\mathbb{E}[Q_i]=R_0$. 
%
%
Note that some of the nodes directly connected to an infectious node (say node $i$) may have already received the information. Hence the number of new infectious nodes created by infectious node $i$ at each transmission is not larger than $Q_i$. To obtain a lower bound on the delay, we consider that the number of new infectious nodes created by an infectious node ($i$) at a transmission is equal to $Q_i$. 

It is straightforward that the total number of infectious nodes at time $k\tau_s$ obeys $a_k=a_{k-1}+\sum_{i=1}^{a_{k-1}}Q_i$, but an explicit form for $a_k$ is difficult to find \cite{harris-the-2002}. To simplify the analysis, an asymptotic model is used. 
Specifically, as we let $L\rightarrow \infty$ while keeping the density $\frac{N}{L^2}$ unchanged, the expected number of children $\mathbb{E}[Q_i]$ does not vary as $N\rightarrow\infty$. 
Further because $a_k$ only depends on $a_{k-1}$ and the distribution of $Q_i$, according to the mean field limit \cite{isham-stochastic-2004}, as $N\rightarrow\infty$, the number of nodes in the $k^{th}$ generation converges almost surely to the  deterministic form
$
a_k= a_{k-1}+a_{k-1}\mathbb{E}[Q_i]
=a_{k-1}\left(1+\frac{R_0}{\beta}\right).
$

Then because $a_0=1$, it follows that for $N\rightarrow\infty$,
$
a_k=\left(1+\frac{R_0}{\beta}\right)^{k} \label{eq_lower_bound_ak}
$ and
\begin{eqnarray}
\lim_{n\rightarrow\infty}T(z)
\geq \mathbb{E}\left [\tau_{s} \arg\max_{k} \left(a_k \leq Nz \right) \right] 
= \tau_{s} \mathbb{E}\left[\arg\max_{k} \left((1+\frac{R_0}{\beta})^{k} \leq Nz \right)\right] 
= \tau_{s} \left\lfloor 1+\frac{\ln{Nz}}{\ln(1+\frac{R_0}{\beta})} \right\rfloor  .
\end{eqnarray}
\end{IEEEproof}

\subsection{Optimization}\label{section_optimization}
In the previous sections, we postulated that the consecutive transmissions of an infectious node are separated by a random time interval $\tau_s+\tauup_r\triangleq\tau$. 
The randomly distributed time interval $\tau$ reflects the uncertainty in the channel access time in practical scenarios. This section disregards the technology limitation (i.e. we assume that channel access time can be pre-determined and there is no contention or collision between concurrent transmissions of different nodes) and investigates the optimal probability distribution of $\tau$ that maximises the effective node degree when all other parameters (i.e. $\eta,\sigma,\lambda,V,r_0$ and $\beta$) are fixed.

Note that this section conducts optimization under the UDM for $\beta=2$ only (where the $R_0$ is equal to the expression given in Eq. \ref{eq_R0}). The same method can be applied to other scenarios which however involve a non-trivial complex analysis (of the intersection area of multiple circles) and hence are left as future work. The performance of the information dissemination process using the optimal broadcast scheme is shown and discussed later in Section \ref{section_RDM_simulation}.

\begin{thm}\label{thm_optimization}
Denote by $p_\tau(\tau)$ the pdf of the random time interval $\tau$. In networks under UDM, consider a set of proposed broadcast schemes with $\beta=2$ but using different sleep strategies, i.e. different values of $\tau_s$ and different distributions for $\tauup_r$. Among all sleep strategies with the same mean sleep time interval $\mathbb{E}[\tau]$, the optimal one that maximises the effective node degree $R_0$ is a strategy using a constant sleep time interval with length $\mathbb{E}[\tau]$.
\end{thm}

\begin{IEEEproof}
Note that under UDM $r_1=r_2=r_0$. Take the second derivative of Eq. \ref{eq_intersection} with regard to $\tau$, there results
\begin{eqnarray}\label{eq_optimization_1}
&&\hspace{-20px}\frac{\partial^2}{\partial \tau^2} A_p(\theta,\tau,r_0,r_0)=
\hspace{-1px}\begin{cases}
0, ~~\mathrm{for}~\psi(\theta, \tau)= 0, \\
\frac{\partial^2  }{\partial \tau^2} 2 r_0^2\arccos(\frac{\psi(\theta, \tau)}{2r_0})
-\frac{\partial^2 \frac{1}{2}\hspace{-1px}\sqrt{[(2r_0)^2-\psi^2(\theta, \tau)][\psi^2(\theta, \tau)]}}{\partial \tau^2},\\
~~~\mathrm{  for ~~ } 0<\psi(\theta, \tau)<2r_0\cr
 0, ~~\mathrm{otherwise}.
\end{cases}
\end{eqnarray}

First consider the second derivative of Eq. \ref{eq_intersection} for $0<a\tau<2r_0$. 
Let $a=2V\sin\frac{\theta}{2}$. Then $\psi(\theta, \tau)=a\tau$. Then there holds
\begin{eqnarray}
\mathcal{A}_1&\triangleq&\frac{\partial^2}{\partial \tau^2} 2 r_0^2\arccos\left(\frac{\psi(\theta, \tau)}{2r_0}\right)  
=\frac{\partial^2}{\partial \tau^2} 2 r_0^2\arccos\left(\frac{a\tau}{2r_0}\right) 
=-\frac{2 a^3 \tau r_0^2}{\sqrt{4r_0^2-a^2 \tau^2}(4r_0^2-a^2 \tau^2)}. 
\end{eqnarray}

Further, 
\begin{eqnarray}
\mathcal{A}_2&\triangleq&\frac{\partial^2}{\partial \tau^2}\frac{1}{2}\hspace{-1px}\sqrt{[(2r_0)^2-\psi^2(\theta, \tau)][\psi^2(\theta, \tau)]} 
=\frac{a^5 \tau^3  -6  a^3 r_0^2 \tau}{(4 r_0^2-a^2 \tau^2) \sqrt{4 r_0^2 -a^2 \tau^2}}.
\end{eqnarray}

Then the second derivative of Eq. \ref{eq_intersection} for $0<a\tau<2r_0$ is given by
\begin{eqnarray}\label{eq_optimization_2}
\mathcal{A}_1-\mathcal{A}_2
=\frac{-2 a^3 \tau r_0^2-a^5 \tau^3  +6  a^3 r_0^2 \tau}{ (4 r_0^2-a^2 \tau^2)\sqrt{4r_0^2-a^2 \tau^2}} 
=a^3 \tau \frac{4 r_0^2-a^2 \tau^2 }{ (4 r_0^2-a^2 \tau^2)\sqrt{4r_0^2-a^2 \tau^2}} 
>0.
\end{eqnarray}

It is evident that $A_p(\theta,\tau,r_0,r_0)$ is a convex function of $\tau$. According to the Jensen's inequality, there holds
$
\mathbb{E} [A_p(\theta,\tau,r_0,r_0)] \geq A_p(\theta,\mathbb{E} [\tau],r_0,r_0).
$

Then according to Lemma \ref{lem_cluster_coefficient}, the clustering factor is
\begin{eqnarray}
\phi(\tau) =  \int_0^{\infty}\int_0^{\pi} A_p(\theta,\tau,r_0,r_0)  \frac{\lambda}{\pi} p_\tau(\tau) d\theta d\tau_r 
\geq \int_0^{\pi} \frac{\lambda}{\pi} A_p(\theta,\mathbb{E} [\tau],r_0,r_0)d\theta.
\end{eqnarray}
 
Finally, according to Theorem \ref{thm_R0}, under UDM when $\beta=2$,  $R_0=2 \lambda\pi r_0^2 -\phi(\tau)$.

This means that among all distributions of the sleep time interval $\tau$ with mean $\mathbb{E} [\tau]$, the case where consecutive transmissions separate by a constant $\mathbb{E} [\tau]$ minimises the clustering factor $\phi(\tau)$, consequently maximises the effective node degree.
\end{IEEEproof}


Note that only Theorem 9 has the limitation
of $\beta=2$, while other results presented in the paper can be applied
to any positive integer value of $\beta$.

\begin{rem}
Note that the result of Theorem
9 only applies to the case where nodes move at the same constant speed.
In the case where nodal speed is randomly distributed, the optimal
sleep interval depends on the speeds of the active nodes as well as
the speeds of other nodes, or more precisely, depends on the relative
speeds between the active nodes and other nodes. Specifically, when
an active node moves slower than V and other nodes move at speed V,
the active node needs to sleep a longer time in order to move away
from its previous transmission location. On the other hand, if the
\textcolor{blue}{active node} moves slower than V but other nodes move faster than
V, then the active node may not necessarily need a longer sleep time
interval because other nodes can move away from the active node. Therefore,
in the case where nodal speed is randomly distributed, there does
not exist a simple optimal sleep time interval that maximises the
effective node degree of every node. An algorithms need to be designed
to adjust the sleep time interval of a node dynamically according
to the relative speeds of the active node and all its neighbours.
\end{rem}

\begin{rem}
The randomly distributed time interval
introduced in Section III-B reflects the uncertainty in the channel
access time in practical scenarios. Theorem 9 disregards this limitation
(i.e. we assume that channel access time can be pre-determined and
there is no contention or collision between concurrent transmissions)
and investigates the optimal probability distribution that maximises
the effective node degree. 
\end{rem}

\section{Simulation results}\label{section_simulation}

\subsection{Random direction mobility simulation}\label{section_RDM_simulation}
Initially, $N=1280$ nodes are uniformly deployed on a torus $(0,800]^2$, i.e. density is $\lambda=0.002$ (nodes/$m^2$), which is the population density in Sydney, Australia \cite{city_population_density}. After deployment of the nodes, they start to move according to the mobility model introduced in Section \ref{section_system_model}. The speed $V$ is set to 10 (m/s) 
. The sleep time interval is a constant, i.e. $\tauup_r=0$, except for Fig.~\ref{pic_optimization}. 
%

Fig. \ref{pic_burst_log_all} shows the percolation probability and the expected fraction of informed nodes. 
It can be seen that both metrics improve as either of $\tau_{s}$ or $\sigma$ increases, owing to the reduction of the clustering factor as shown in Section \ref{section_efficiency}. 
Our analytical bounds are close to the simulation results as shown in Fig. \ref{pic_burst_log_all}(a), while the discrepancy in Fig. \ref{pic_burst_log_all}(b) is caused by the well-mixing assumption (c.f. \cite{zhang-on-2011-manet}) used in deriving the fraction of informed nodes, which however requires another non-trivial analysis to adjust.

\begin{figure}[ht]
\centering\vspace{-5px}
\includegraphics[scale=0.6]{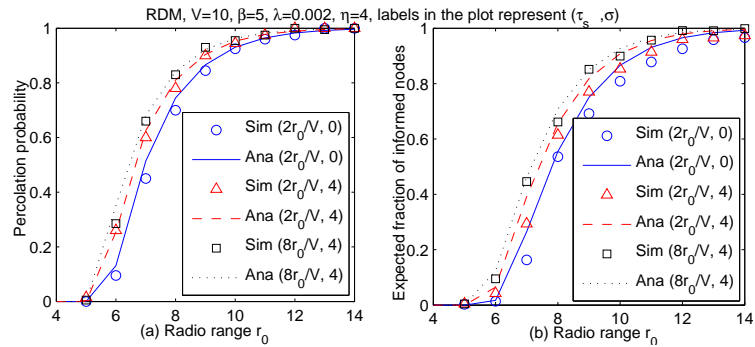}\vspace{-6px}
\caption{Analytical (Ana) and simulation (Sim) results of the percolation probability and expected fraction of informed nodes, using different network parameters. Analytical results are the upper bounds obtained by combining Theorem \ref{thm_percolation}, \ref{thm_size} and \ref{thm_R0}. Simulation result for percolation probability shows the probability of having at least 10\% informed nodes in the steady state.}\vspace{-5px}
\label{pic_burst_log_all}
\end{figure}

Note that the percolation probability defined
in Definition 3 is an asymptotic figure of merit, whereas any simulation
can only be conducted on a finite network. The simulation in this
paper is conducted in a finite but large network with 1280 nodes.
The simulation result for percolation probability shows the probability
of having at least $x=10$\% informed nodes in the steady state. Note
that we have tried some different values for $x$ ranging from 5\% to
20\% and the plots are similar. It can be seen that the percolation model, though constructed for an infinite network, is of predictive value for networks likely to be encountered in real world.

Fig. \ref{pic_burst_rayleigh} illustrates the impact of Rayleigh fading on the percolation probability and the expected fraction of informed nodes. It can be seen the Rayleigh fading has a negative impact on the information dissemination process, which can also be seen from the analysis (c.f.  Lemma \ref{lem_cluster_coefficient_mono}).

\begin{figure}[ht]
\centering\vspace{-5px}
\includegraphics[scale=0.6]{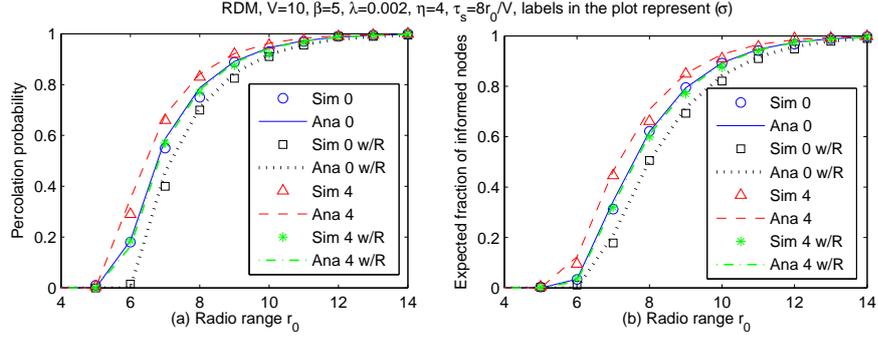}\vspace{-6px}
\caption{The percolation probability and expected fraction of informed nodes, using different network parameters, where ``w/R'' represents ``with Rayleigh fading''. Note that the curves of ``Ana 0'' and ``Ana 4 w/R'' almost overlap hence hard to distinguish.}\vspace{-6px}
\label{pic_burst_rayleigh}
\end{figure}

\begin{figure}[ht]
\centering\vspace{-5px}
\includegraphics[scale=0.65]{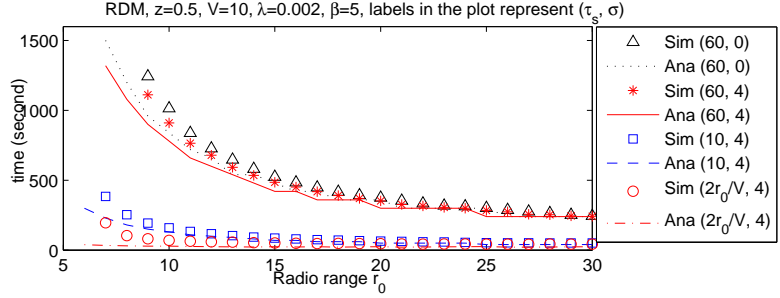}\vspace{-6px}
\caption{The lower bound on the delay for a piece of information to be received by $50\%$ of nodes using different network parameters. For the sake of comparison, the value of $\tau_{s}$ is kept constant (i.e. $10$ or $60$ seconds) while the value of $r_0$ being increased in the first six curves.}\vspace{-6px}
\label{pic_burst_delay_lower}
\end{figure}

Fig. \ref{pic_burst_delay_lower} shows the lower bound of the delay for a piece of information to be received by 50\% of nodes. It can be seen that the length of the sleep time interval has a great impact on the delay and our analytical results provide valid lower bounds on the delay. It is clear that a longer sleep time interval can cause a longer delay. A network designer needs to consider the trade-off between delay and resource consumption.

Further, Fig. \ref{pic_burst_log_all}, Fig \ref{pic_burst_rayleigh} and Fig. \ref{pic_burst_delay_lower} suggest that shadowing effects benefit the information dissemination process in terms of percolation probability, expected fraction of informed nodes and delay, because an increase in $\sigma$ leads to an increase in the effective node degree $R_0$ as shown in Theorem \ref{thm_R0}. This is in sharp contrast with previous conclusions, e.g. \cite{qin-on-2003}, because traditional routing algorithms like AODV need to establish a route before sending data, whilst the route is unstable due to dynamic topology and channel randomness.

Moreover, the proposed scheme and the analysis presented in the paper can also be useful for uni-cast. For example, consider that a source node sends a uni-cast packet to a randomly-chosen destination node using the proposed broadcast scheme. Then it is interesting to note that the probability that the packet reaches the randomly-chosen destination at time $T(z)$ is $z$ (c.f. Theorem \ref{thm_delay_lower}). Moreover, the probability that the packet reaches the randomly-chosen destination at the steady state is given by the expected fraction of informed node $z_0$ (c.f. Theorem \ref{thm_size}).

\begin{figure}[ht]
\centering\vspace{-5px}
\includegraphics[scale=0.7]{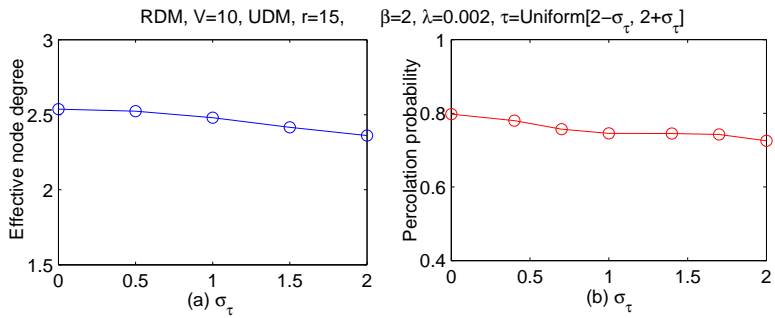}\vspace{-6px}
\caption{Simulation results of (a) the effective node degree and (b) the percolation probability in networks using the proposed broadcast scheme where the time interval between consecutive transmissions of an infectious node is uniformly distributed in $[2-\sigma_\tau,2+\sigma_\tau]$, where $\sigma_{\tau}$ is varied from 0 to 2.}\vspace{-5px}
\label{pic_optimization}
\end{figure}

Fig. \ref{pic_optimization} (a) verifies the results of Theorem \ref{thm_optimization} that a constant time interval between consecutive transmissions leads to the largest effective node degree, which further leads to the largest percolation probability as shown in Fig. \ref{pic_optimization} (b). Note that in a practical scenario where concurrent transmissions of different nodes can cause wireless channel contention and transmission failure, a further analysis is required to characterize the effective node degree and the percolation probability, which is left as future work.

\subsection{Energy and bandwidth efficiency
simulation results}

This subsection presents the simulation results
evaluating the energy and bandwidth efficiency of the proposed scheme.
As described in Section IV-E, the consumption of energy and bandwidth
for the dissemination of a piece of information is an increasing function
of the number of transmissions for the packets containing the information.
According to a well-known study by Feeney et al. in \cite{feeney-investigating-2001},
using 11Mbps wireless transmission, the energy consumed by a node
to broadcast a packet of size 1K bytes is 482$\mu$J, and the energy
consumed by a node to receive a broadcast packet of size 1K bytes
is 76$\mu$J. Therefore, the overall energy consumption, including
the transmitting and receiving energy, is $E_{c}=482+\mathbb{E}[\mathcal{D}_{i}]76$,
where $\mathbb{E}[\mathcal{D}_{i}]$ is the average node degree, which
is straightforwardly equal to the right hand side of Eq. \ref{eq_R0} (i.e. it is no longer an upper bound) by letting
$\beta=1$.

Fig. \ref{pic_no_transmissions1_energy} shows the simulation results of
the overall energy consumption for the dissemination of a piece of
information using the proposed efficient scheme and the traditional
SIR scheme. Firstly, it can be seem that the proposed scheme has a
significantly smaller energy consumption compared with the traditional
SIR scheme. Further, the increase in the density of nodes has much
less impact on the overall energy consumption than the SIR scheme.
This is because a larger nodal density causes more overlaps between
consecutive transmissions, which leads to a larger clustering factor
and consequently more energy consumption for the traditional SIR scheme.
On the other hand, the enforced gaps between transmissions under the
proposed scheme reduce the aforementioned overlaps hence the proposed
scheme is more energy efficient.

\begin{figure}[ht]
\centering
\includegraphics[scale=0.6]{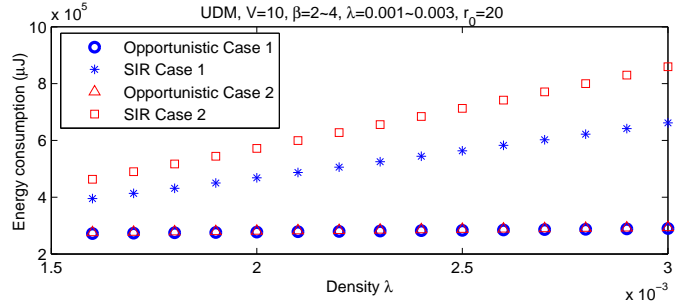}\vspace{-6px}
\caption{Simulation results of the overall energy
consumption using the proposed scheme and the traditional SIR scheme
for a packet to be received by 99\% of nodes in a network. In Case
1 and Case 2, $\beta=2$ and 4 respectively.
}\vspace{-5px}
\label{pic_no_transmissions1_energy}
\end{figure}

Fig. \ref{pic_no_transmissions2_energy} evaluates the impact of the sleep
time interval $\tau_{s}$ on the energy consumption. Note that other
things being equal, the energy and bandwidth consumption is a linearly
increasing function of the number of transmissions. It can be seen
that increasing the sleep time interval reduces the energy consumption.
This is because a larger gap between transmissions leads to less overlaps
and the same reason applies to the network under the unit disk model
(labeled $\sigma=0$), the network under log-normal shadowing model
(labeled $\sigma=4$) and the networks subject to Rayleigh fading
(labeled \emph{with Rayleigh}).

\begin{figure}[ht]
\centering
\includegraphics[scale=0.7]{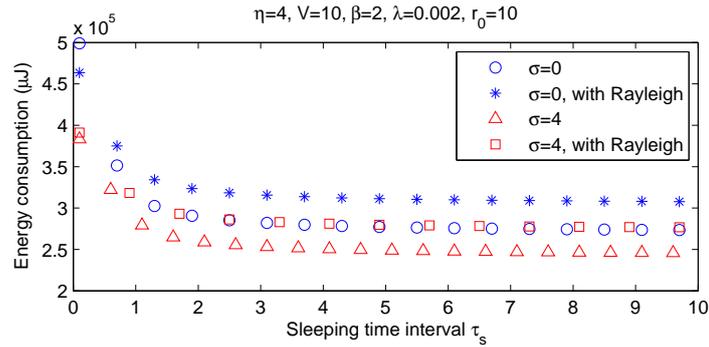}\vspace{-6px}
\caption{Simulation results of the overall energy
consumption using the proposed scheme and the traditional SIR scheme
for a packet to be received by 99\% of nodes is a network. The impact
of the sleep time interval $\tau_{s}$ on the overall energy consumption
is evaluated by letting the random component $\tau_{r}$ take a constant
value 0.
}\vspace{-5px}
\label{pic_no_transmissions2_energy}
\end{figure}

Fig. \ref{pic_no_transmissions3_energy} compares the proposed scheme with
two classic schemes in the literature: the SI scheme and the SIR scheme.
As introduced in Section II, a large number of studies in this area
considered the Susceptible-Infectious (SI) scheme, where every node
carries the received information and forwards it to all nodes coming
into the radio range. In particular, Thedinger et al. proposed a SI-based
broadcast scheme in \cite{thedinger-store-2010}, where each node
broadcasts the received information to its directly-connected neighbours
repeatedly and the consecutive re-broadcasts are separated by a fixed
time interval. Using a SI-based scheme, every node needs to re-broadcast
the received packet. For a fair comparison, we stop the re-broadcast when the packet is received by 99\% of nodes in the
network. It can be seen in Figure 3 that the SIR-based model requires
a significantly smaller overall energy consumption compared with SI-based
model. This is because the recovery mechanism in the SIR model limits
the number of transmission of each node to save energy. Moreover,
as described earlier, the scheme proposed in this paper further reduces
the energy consumption by reducing the overlaps between consecutive
transmissions and thereby reducing the number of transmissions to
nodes already having a copy of the packet. Therefore it can be seen
that the proposed scheme requires less number of transmissions for
a packet to be received by the same percentage of nodes in a network.

\begin{figure}[ht]
\centering
\includegraphics[scale=0.6]{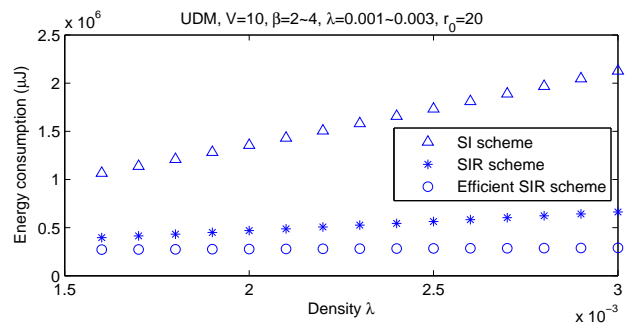}\vspace{-6px}
\caption{Simulation results of the overall energy
consumption using the proposed scheme and the other classic schemes
for a packet to be received by 99\% of nodes in a network.
}\vspace{-6px}
\label{pic_no_transmissions3_energy}
\end{figure}

\subsection{Real mobility trace simulation}\label{section_trace_simulation}
In this subsection, we consider a MANET driven by a real world trace of the cabs in the San Francisco Bay Area \cite{epfl_trace}. 
%
The cab trace contains GPS coordinates of 536 cabs over approximately 30 days. Linear interpolation was used to increase the trace granularity. 
We focus on the downtown area (3.5km$\times$3.5km) as indicated by the square in Fig.~\ref{pic_trace_area}. 
A piece of information is broadcast from an arbitrary cab in the downtown area using the scheme proposed in Section~\ref{section_system_model}. When information dissemination stops, we count the fraction of cabs in the downtown area that have received the information.  

It is obvious that the taxis do not follow
the mobility model used in the analysis. Actually there has been no
analytical model that can fully characterize the movement of real
world devices. The analysis based on analytical models under some idealizing assumptions is used to predict the real-world situations. We use the cab trace to evaluate the analytical results
in a large scale network when the node mobility and distribution deviate
from assumptions described in Section III, and to identify the limitations
of the analysis and present possible usage in estimating the performance
of information dissemination in real world dynamic networks. It is
interesting to see in the following figures that the impacts of varying
each network parameters on the network performance predicted by the
analysis highly coincide with the simulations results.

\begin{figure}[ht]
\centering
\includegraphics[scale=0.39]{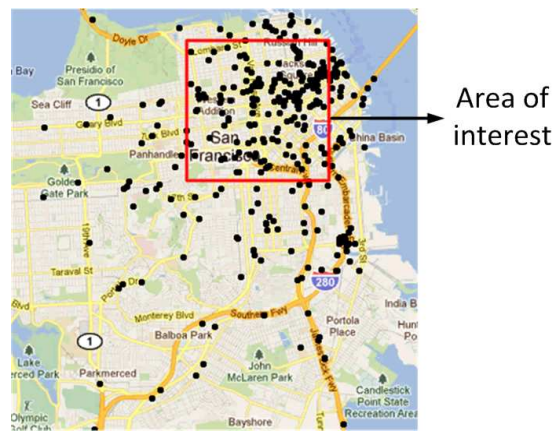}\vspace{-6px}
\caption{A snapshot of the San Francisco Bay Area. The dots represent the positions of the cabs at a particular time instant. The square represents the downtown area we are interested in. 
}\vspace{-8px}
\label{pic_trace_area}
\end{figure}

\begin{figure}[ht]
\centering\vspace{-5px}
\includegraphics[scale=0.7]{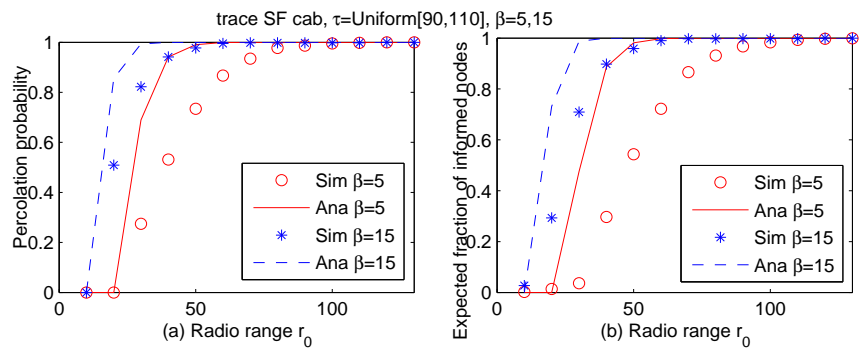}\vspace{-6px}
\caption{Percolation probability and the expected fraction of informed nodes for the MANET driven by cab trace.}\vspace{-8px}
\label{pic_all_trace}
\end{figure}

Fig. \ref{pic_all_trace}(a)(b) shows the percolation probability and the expected fraction of informed nodes in the steady state. 
%
%
It can be seen that though the analytical result and simulation result exhibit the same trend, the fraction of informed nodes in a MANET driven by the actual cab trace is lower than that predicted by analysis. The discrepancy is mainly caused by the inhomogeneity of the node distribution. Specifically, as can be seen in Fig. \ref{pic_trace_area}, the density of cabs in suburbs is lower than that in the downtown area. Therefore when an infectious node moves to the suburbs, it has little chance to inform other nodes. 

\section{Conclusion and future work}\label{section_conclusion}
This paper proposed an energy and bandwidth efficient broadcast scheme for MANETs subject to channel randomness. The proposed scheme is decentralized and simple to implement. Further, the performance of the network using the proposed scheme, measured by the percolation probability, expected fraction of informed nodes and delay, was analytically studied. The accuracy of the analytical results was verified by simulations.

In the future, we are going to investigate the performance of broadcast schemes where the distribution of $\tauup_r$ is determined by certain media access control protocols (e.g. CSMA). It is also interesting to investigate the throughput capacity of networks using the proposed  scheme, which is also expected to be better than that of networks using traditional epidemic schemes. Moreover, the analysis presented in this
work can be extended to consider the networks with infrastructure
support. More specifically, a multi-type branching process \cite{harris-the-2002}
can be used to replace the Galton-Watson branching process in Theorem
5 to extend this work to heterogeneous networks. 
\vspace{-5px}

\bibliographystyle{IEEEtran}
\bibliography{endnote}

\end{document}